\documentclass[submission,copyright,creativecommons]{eptcs}
\providecommand{\event}{AiML 2026} 

\usepackage{aiml26}
\usepackage{iftex}

\usepackage{graphicx}
\usepackage{amssymb}
\usepackage{xcolor}

\ifpdf
\usepackage{underscore}         
\usepackage[T1]{fontenc}        
\else
\usepackage{breakurl}           
\fi

\definecolor{darkgreen}{RGB}{20,120,20}
\definecolor{newgreen}{RGB}{40,160,70}

\title{Labelled Sequents for Inquisitive First-Order Modal Logic}
\author{Ivano Ciardelli
	\institute{Dipartimento FISPPA\\ University of Padua, Italy}
	\email{ivano.ciardelli@unipd.it}
	\and
	Simone Conti
	\institute{Dipartimento FISPPA\\ University of Padua, Italy}
	\email{simone.conti@phd.unipd.it}
}

\newcommand{\titlerunning}{Labelled Sequents for Inquisitive First-Order Modal Logic}
\newcommand{\authorrunning}{I. Ciardelli, S. Conti}

\hypersetup{
	bookmarksnumbered,
	pdftitle    = {\titlerunning},
	pdfauthor   = {\authorrunning},
	pdfsubject  = {EPTCS},               
	pdfkeywords = {inquisitive modal logic, labelled sequent calculus, modal predicate logic, global supervenience} 
}

\newcommand{\inqbq}{\textsf{InqBQ}}
\newcommand{\inqwq}{\textsf{InqWQ}}
\newcommand{\inqwqml}{\ensuremath{\textsf{InqQML}^{-}_{\Box}}}

\newcommand{\lori}{\mathbin{\rotatebox[origin=c]{-90}{$\geqslant$}}}

\newcommand{\existsi}{\mathord{\exists\hspace{-.4em}\exists}} 
\newcommand{\quotes}[1]{``#1''}
\renewcommand{\phi}{\varphi}

\newcommand{\s}{\textsf{s}}

\usepackage{amsmath}
\usepackage{proof}
\usepackage{bussproofs}
\usepackage{mathtools}
\usepackage{enumitem}
\usepackage{wasysym}
\usepackage{float}
\newcommand{\N}{\mathbb{N}}
\renewcommand{\t}{\textsf{t}}
\newcommand{\w}{\textsf{w}}
\newcommand{\vv}{\textsf{v}}
\newcommand{\R}{\textsf{R}}
\newcommand{\W}{\textsf{W}}
\renewcommand{\u}{\textsf{u}}
\newcommand{\IWMC}{\textnormal{\textsf{IWMC}}}

\newcommand{\To}{\Rightarrow}
\newcommand{\nmodels}{\not\models}
\newtheoremstyle{named}{}{}{\itshape}{}{\bfseries}{.}{.5em}{#1 \thmnote{#3}}
\theoremstyle{named}
\newtheorem*{proposition*}{\textbf{Proposition}}
\newtheorem*{lemma*}{\textbf{Lemma}}

\usepackage[normalem]{ulem}
\begin{document}
	\maketitle
	
	\begin{abstract}
		In recent work \cite{Ciardelli:25}, an inquisitive first-order modal logic has been proposed to reason about relations of modal dependence, including the notion of global supervenience (functional dependence among the extensions of predicates relative to a space of possibilities). 
		At present, no proof system exists for this logic. 
		We provide a complete labelled sequent calculus, extending a calculus developed by Litak and Sano \cite{LitakSano:25} for a weak version of inquisitive first-order logic. We prove strong completeness for the calculus and show that it enjoys desirable structural properties, including the invertibility of its rules and the admissibility of cut. 
	\end{abstract}

	\section{Introduction}
	\label{sec:intro}
	
	Inquisitive logics \cite{Ciardelli:18book,Ciardelli:23book,Puncochar:16generalization,Puncochar:19} extend existing versions of classical or non-classical logics with formulas representing questions. This extension is typically achieved by lifting the semantics from standard points of evaluations (e.g., valuation functions, relational models, or possible worlds) to sets of such points, called \emph{information states} (or just \emph{states} for short), relative to which a notion of \emph{support} is defined.\footnote{An analogous move is made in the \emph{team semantics} tradition, where assignments are replaced by sets of assignments called \emph{teams} (see, a.o., \cite{Hodges:97,Vaananen:07,Galliani:12,GradelVaananen:13}). For connections between these two lines of work see, among others, \cite{YangVaananen:16,Ciardelli:16dependency}.}
	
	The standard system of inquisitive first-order logic, \inqbq, extends classical first-order logic with two question-forming operators: inquisitive disjunction, $\lori$, and the inquisitive existential quantifier, $\existsi$. In terms of $\lori$, a polar question operator $?$ is defined by letting ${?}\phi:=(\phi\lori\neg\phi)$. In spite of important advances \cite{Grilletti:17,Grilletti:21,GrillettiCiardelli:23,CiardelliGrilletti:22}, the meta-theoretic properties of this system remain rather mysterious: in particular, it is not known whether \inqbq\ is recursively axiomatizable, or whether it is entailment-compact.\footnote{A logic is said to be entailment-compact if whenever a formula $\psi$ follows from a set of formulas $\Phi$, it follows from a finite subset $\Phi_0\subseteq\Phi$. In inquisitive logic, this is not equivalent to the more common formulation of compactness in terms of satisfiability, essentially due to the fact that the double negation law does not hold in general.}
	
	Recent work \cite{CiardelliGrilletti:22,Conti25} has focused on the $\existsi$-free fragment of this logic, denoted \inqwq. This fragment is sufficiently expressive to regiment a broad variety of questions; for instance, in addition to standard sentences like $\forall xPx$, regimenting the statement \quotes{every object is $P$}, \inqwq\ also contains sentences like ${?}\forall xPx$ and $\forall x{?}Px$, regimenting, respectively, the questions \quotes{whether or not every object is $P$} and \quotes{which objects are $P$}. 
	At the same time, \inqwq\ enjoys a strong semantic property which is not shared by all of \inqbq: \emph{finite coherence}.\footnote{The notion of coherence goes back to the team semantics literature, where it was first investigated by Jarmo Kontinen \cite{Kontinen:13}.} A formula $\phi$ is said to be $n$-coherent for a number $n\in \mathbb{N}$ if, in order to decide whether $\phi$ is supported by a state $s$, it suffices to check whether it is supported by the subsets $s'\subseteq s$ of cardinality up to $n$. As shown by Ciardelli and Grilletti \cite{CiardelliGrilletti:22}, every formula $\phi$ in \inqwq\ is $n$-coherent for some number $n$ which is computable from $\phi$. This fact has far-reaching repercussions: one can use it to show, among other things, that \inqwq\ is entailment-compact and recursively axiomatizable \cite{CiardelliGrilletti:22}. 
	
	In recent work, two complete proof systems have been provided for \inqwq: Conti \cite{Conti25}, building on \cite{CiardelliGrilletti:22}, gave a natural deduction system, while Litak and Sano \cite{LitakSano:25} gave a labelled sequent calculus. 
	Both systems rely crucially on the finite coherence property of \inqwq. 
	In the former system, this property shows up through a specific \emph{coherence rule} which allows one to discharge certain assumptions about the cardinality of the state of evaluations. In the latter system, coherence ensures that it is sufficient to work with labels which are finite sequences of indices, representing finite states. 
	
	In recent work, Ciardelli \cite{Ciardelli:25} investigated a modal logic \inqwqml\ obtained by extending \inqwq\ with a generalization of the Kripke modality $\Box$. The motivation for this work came from the analysis of \emph{global supervenience}, a notion of dependence that has received attention in the philosophical literature \cite{Kim:84,Stalnaker:96,McLaughlin:97,Bennett:04,Leuenberger:09}, where it has been invoked to give a precise formulation of certain philosophical theses, such as David Lewis's thesis of \emph{Humean supervenience} \cite{Lewis:86}. Intuitively, global supervenience captures the idea that the overall distribution of certain properties or relations in the world is fully determined by the overall distribution of other properties or relations. An example: if we fix who is a parent of whom, we thereby also fix who is a grandparent of whom; so, the \emph{grandparent-of} relation globally supervenes on the \emph{parent-of} relation.
	Formally, we may say that a predicate $Q$ globally supervenes on another predicate $P$ at a possible world $w$ if any two successors of $w$ which assign the same extension to $P$ also assign the same extension to $Q$. (The notion extends straightforwardly to the case of several predicates supervening on several other predicates.)
	Thus, global supervenience captures functional dependencies between the extensions of predicates across a space of possibilities.
	As Ciardelli \cite{Ciardelli:25} showed, global supervenience claims cannot be expressed in standard quantified modal logic, but they \emph{can} be expressed in \inqwqml\ in a particularly perspicuous way, namely, as strict conditionals whose antecedents and consequents are questions. Thus, e.g., the claim that $Q$ globally supervenes on $P$ is formalized by the strict conditional
	$$\Box(\forall x{?}Px\to\forall x{?}Qx)$$
	having as its antecedent the subvenient question $\forall x{?}Px$ (\quotes{which objects are $P$}), and as its consequent the supervenient question $\forall x{?}Qx$ (\quotes{which objects are $Q$}). As discussed in \cite{Ciardelli:25}, this analysis is insightful, since it allows us to trace back logical properties of global supervenience to familiar properties of strict conditionals and questions. The inquisitive modal logic \inqwqml, then, provides an attractive system to reason about global supervenience claims, as well as modal dependence claims more generally \cite{Ciardelli:18aiml}. 
	
	In view of this motivation, it seems important to have a proof system for this logic: such a system would allow one to formally establish the validity of certain inferences concerning modal dependencies and, ideally, it could be used to get further insight about the logic of these dependence notions. 
	While Ciardelli \cite{Ciardelli:25} proved (by means of a translation to two-sorted first-order logic) that the set of validities of \inqwqml\ is recursively enumerable, he left the development of a proof system as an open problem.
	
	In this paper, we fill this gap, providing a labelled sequent calculus for \inqwqml. Our calculus builds on Litak and Sano's calculus for \inqwq\  \cite{LitakSano:25}, where labels are finite sequences $\textsf{w}_1\dots \textsf{w}_n$ of indices;\footnote{One may wonder if it is also possible to obtain a proof system for \inqwqml\ by extending the natural deduction system in~\cite{Conti25}. This seems difficult, since the completeness proof in \cite{Conti25} relies on the finite model property of \inqwq, which does not extend to \inqwqml. Regardless of this problem, the labelled approach of \cite{LitakSano:25} seems preferable due to its better proof-theoretic properties.} intuitively, indices represent worlds, and so, labels represent finite states. To this calculus we add rules for the modality $\Box$. These rules are inspired by the standard rules for $\Box$ in labelled sequent calculi \cite{Negri:05} but, at the same time, they use in a crucial way the finite coherence property of \inqwqml, inherited from \inqwq: every formula $\phi$ in \inqwqml\ is $n_\phi$-coherent for some number $n_\phi$ computable from $\phi$ \cite{Ciardelli:25}. 
	This is crucial, in particular, for our right rule for $\Box$. 
	Semantically, $\Box\phi$ is true at a world $w$ if $\phi$ is supported by the state consisting of all the successors of $w$; however, by finite coherence, this reduces to $\phi$ being supported by all sets consisting of at most $n_\phi$-many successors of $w$. This fact allows us to formulate a right rule for $\Box$ that, in essence, says the following: in order to prove a labelled formula $\textsf{w}:\Box\phi$, introduce $n_\phi$-many fresh indices $\textsf{v}_1\dots \textsf{v}_{n_\phi}$ standing for successors of $\textsf{w}$, and aim to prove that $\textsf{v}_1\dots \textsf{v}_{n_\phi}:\Box\phi$.

	In addition to establishing the soundness and completeness of this system, we prove the invertibility of its rules, and the admissibility of the rules of weakening, contraction, and cut.
	
	Inquisitive modal logic is currently an active and rapidly growing field. While several proof systems have been developed for  \emph{propositional} inquisitive modal logics, including labelled sequent calculi \cite{Muller:24,Muller:26}, our calculus represents the first example of a proof system for a \emph{first-order} inquisitive modal logic. We think that, in addition to its direct relevance as a tool to study the logic \inqwqml, our contribution might serve as a model to build proof systems for other first-order inquisitive modal logics.
	
	The paper is structured as follows: Section 2 provides the relevant background on the logic \inqwqml; Section 3 presents our contribution; Section 4 outlines directions for future work.
	
	\section{Background}\label{Section:background}
	In this section, we provide the necessary technical background on the inquisitive modal logic \inqwqml. For a thorough presentation of this logic and proofs of the facts mentioned in this section, we refer to \cite{Ciardelli:25}. For a more general introduction to inquisitive logic, see \cite{Ciardelli:23book}.
	\paragraph{Syntax.}
	We start with a countable first-order signature $\Sigma$. For simplicity, we focus on the case where $\Sigma$ is a set of predicates (each having an associated \textit{arity}), and doesn't contain function symbols nor the identity symbol. The language of \inqwqml{} is given by the following definition, where $P\in\Sigma$ is an $n$-ary predicate and $x,x_1,\dots,x_n$ range over $Var$, a countably infinite set of variables:
	$$\phi\Coloneqq\bot\mid P(x_1,\dots,x_n)\mid\phi\land\phi\mid\phi\lori\phi\mid\phi\to\phi\mid\forall x\phi\mid\Box\phi$$
	We also consider the following defined symbols: $\lnot\phi=\phi\to\bot$, $\phi\lor\psi=\lnot(\lnot\phi\land\lnot\psi)$, $\exists x\phi=\lnot\forall x\lnot\phi$ and the inquisitive operator ${?}\phi=\phi\lori\lnot\phi$.
	We call formulas where $\lori$ does not appear \textit{classical formulas}. These formulas can be identified with those of standard modal logic, with a particular choice of primitive operators. The operator $\lori$ is called \textit{inquisitive disjunction}; informally, it is intended to formalize disjunctive questions. Thus, for instance, whereas the classical disjunction $P(x)\lor\neg P(x)$ regiments the (tautological) statement that \emph{either $x$ is $P$ or it isn't}, the inquisitive disjunction $P(x)\lori\lnot P(x)$ regiments the question \emph{whether $x$ is $P$ or not} (which explains why this formula is abbreviated as ${?}P(x)$).

	\paragraph{Semantics.}
	Models for \inqwqml{} are defined as tuples $M=\langle W,D,R,I\rangle$, where $W\neq\emptyset$ is a set of possible worlds, $D\neq\emptyset$ is a domain of individuals, $R\subseteq W\times W$ is an accessibility relation, associating each world $w\in W$ with a set of \emph{successors} $R[w]=\{v\in W\mid w R v\}$, and $I$ maps each $w\in W$ to an interpretation function that assigns to each $n$-ary predicate $P$ a corresponding extension $I_w(P)\subseteq D^n$. 
	An \textit{information state} (or, simply, a \textit{state}) is a subset of $W$. As usual, an assignment is a function $g:Var\to D$. Given an assignment $g$, a variable $x$, and an individual $d\in D$, $g[x\mapsto d]$ is the assignment that maps $x$ to $d$ and agrees with $g$ on all other variables.

	In \inqwqml{}, formulas $\phi$ are evaluated in terms of a relation of support, written as $M,s\models_g\phi$, relative to a model $M$, a state $s$ and an assignment $g$. We define support conditions inductively as follows:
	\begin{itemize}
		\item $M,s\models_g \bot \iff s=\emptyset$
		\item $M,s\models_g P(x_1,...,x_n)\iff \text{for all }w\in s,\ \langle g(x_1),...,g(x_n)\rangle\in I_w(P) $
		\item $M,s\models_g \varphi\land\psi \iff M,s\models_g \varphi \text{ and } M,s\models_g\psi$
		\item $M,s\models_g \varphi\lori\psi\iff M,s\models_g\varphi \text{ or } M,s\models_g\psi$
		\item $M,s\models_g\varphi\rightarrow\psi\iff \text{for all }t\subseteq s,\ M,t\models_g \varphi \text{ implies }M,t\models_g\psi $
		\item $M,s\models_g \forall x \varphi\iff \text{for all }d\in D,\ M,s\models_{g[x\mapsto d]}\varphi$
		\item $M,s\models_g \Box\phi \iff\text{for all }w\in s:\ M,R[w]\models_g \phi$
	\end{itemize}
	

\noindent	
Entailment is defined as preservation of support: for a set of formulas $\Phi\cup\{\psi\}$, $\Phi\models\psi$ if for all models $M$, states $s$ and assignments $g$, $M,s\models_g\phi$ for all $\phi\in\Phi$ implies $M,s\models_g\psi$. Additionally, we define a restriction of entailment to states $s$ of cardinality at most $n$ (in symbols, $\#s\le n$):  $\Phi\models_n\psi$ if for all models $M$, states $s$ such that $\# s\leq n$ and assignments $g$, $M,s\models_g\phi$ for all $\phi\in\Phi$ implies $M,s\models_g\psi$.

	The above semantics for \inqwqml{} satisfies the following standard properties of inquisitive logics:
	\begin{itemize}
		\item Persistency: if $M,s\models_g\phi$, then for all $t\subseteq s$, $M,t\models_g\phi$
		\item Empty state property: $M,\emptyset\models_g\phi$ for any formula $\phi$
	\end{itemize}
	
	\noindent
	Another important desideratum of inquisitive extensions is conservativity over the original system being extended. We say that a formula $\phi$ is true at a world $w$ (under $g$) if $\{w\}\models_g\phi$. It is easy to check that, for classical formulas, the truth conditions delivered by our definition coincide with those given by standard Kripke semantics. Moreover, classical formulas are \textit{truth-conditional}, meaning that they are supported by a state iff they are true at each of its possible worlds. Using these facts, it is easy to see that entailment among classical formulas in \inqwqml\ coincides with entailment in standard constant-domain quantified modal logic. \inqwqml{} can thus be seen as a conservative extension of the latter logic. Note, furthermore, that the support clause for $\Box$ implies that $\Box\phi$ is always truth-conditional for any $\phi$.
	
\paragraph{Coherence.}	A key property of formulas in \inqwqml\ is \emph{finite coherence} \cite{Kontinen:13,CiardelliGrilletti:22}, defined as follows.
	
	\begin{definition}[$n$-coherence] Given a state $t$, we denote its cardinality by $\#t$. 
		We say that a formula $\phi$ is $n$-coherent for $n\in\N$ if for all $M,s,g$:
		$$M,s\models_g\phi\iff\text{ for all $t\subseteq s$ such that $\# t\leq n$, $M,t\models_g\phi$}$$
	\end{definition}
	
	\noindent
	Intuitively, $n$-coherence of $\phi$ means that in order to check if $\phi$ is supported at a state, it suffices to check if it is supported by all substates of cardinality at most $n$. Crucially, this implies that if $\phi$ is \emph{not} supported at a state, then there is some finite, at most $n$-sized substate that doesn't support $\phi$.
	Notice that, in particular, $1$-coherence coincides with being truth-conditional.
	\begin{proposition}[\protect{\cite[Prop. 6.3]{Ciardelli:25}}]\label{Finite-coherence}
		For every formula $\phi$ of \inqwqml{} there is a number $n_\phi\in\N$, computable from $\phi$, such that $\phi$ is $n_\phi$-coherent. In particular, $n_\phi$ can be computed inductively as follows: 
		$n_\varphi=1$ for any atomic formula; $n_{(\psi\land\xi)}=max\{n_{\psi},n_{\xi}\}$; $n_{(\psi{\footnotesize\lori}\xi)}=n_\psi+n_\xi$; $n_{(\psi\rightarrow\xi)}=n_{\xi}$; $n_{(\forall x_i \psi)}=n_{\psi}$ and $n_{\Box\varphi}=1$. 
	\end{proposition}
	\noindent	
	Importantly, combining this result with persistency implies that the validity of any \inqwqml{} entailment can be reduced to its validity over a class of finite states with bounded cardinality.
	\begin{proposition}\label{Entailment-finite-coherence}
		For any set of \inqwqml{} formulas $\Phi\cup\{\psi\}$, $\Phi\models\psi\iff\Phi\models_{n_\psi}\psi$. 
	\end{proposition}
	%
	
	\noindent
	Using this fact, one can show that \inqwqml\ is \emph{entailment-compact}, in the following sense.
	\begin{proposition}[Entailment-compactness, \protect{\cite[Prop. 6.8]{Ciardelli:25}}]\label{compactness}
		For any set of \inqwqml{} formulas $\Phi\cup\{\psi\}$, if $\Phi\models\psi$, then $\Phi_0\models\psi$  for some finite $\Phi_0\subseteq\Phi$.
	\end{proposition}

	\section{Labelled sequent calculus and completeness proof}\label{Section-Results}
	In this section, we describe a labelled sequent calculus called \IWMC{} (for inquisitive weak modal calculus) and prove its soundness and strong completeness for \inqwqml{}.  
%

	Intuitively, labels represent finite states, and labelled formulas translate the support relation. Formally, labels are non-empty finite sets of indices, where each index is a natural number.
	\footnote{We exclude empty labels for simplicity. While the empty state is officially allowed in the semantics of \inqwqml{}, it plays a trivial role, since it supports any formula whatsoever, and can be omitted without affecting the logic.} We use the meta-variables $\w,\textsf{v},\textsf{u}$ for indices, and $\s,\t,\u,\dots$ for labels. We consider two kinds of expressions: 
	\begin{itemize}
		\item Labelled formulas of the form $\s:\phi$, where $\phi$ is a formula and $\s\subseteq_{\text{fin}}\N$ 
		\item Relational atoms of the form $\w\R\vv$ where $\w,\vv\in\N$.
	\end{itemize}
	\noindent
	As suggested by the notation, the former intuitively represent support at a state, while the latter syntactically encode instances of the accessibility relation. 
A sequent is an ordered pair $\Gamma\To\Delta$ consisting of finite multisets $\Gamma$ and $\Delta$, where $\Gamma$ can contain labelled formulas and relational atoms, while $\Delta$ contains only labelled formulas.

Below, we provide the rules of the sequent calculus \IWMC. We write $\w\R\t$ for the set $\{\w\R\vv\mid \vv\in\t\}$. The universe $\W_{\Gamma,\Delta}$ of a sequent $\Gamma\To\Delta$ is the set of indices that occur in it, within labels or in relational atoms.


\newcommand{\rulename}[1]{\textnormal{\textsf{$(#1)$}}}

\begin{figure}[H]
	\centering
	{\large The Sequent Calculus \IWMC\\[.5cm]}\par\smallskip

	\setlength{\arraycolsep}{1.6em}
	\renewcommand{\arraystretch}{1.35}
	$\begin{array}{@{}c@{\qquad}c@{}}
		\infer[\rulename{\text{\textsf{id}}}\;\text{\small where } \s\supseteq \t]
		{\s:P(\bar x),\,\Gamma \To \Delta,\,\t:P(\bar x)}
		{}
		&
		\infer[\rulename{\bot{\To}}]{\s:\bot,\,\Gamma \To \Delta}{}\\[2.4ex]
		
		\multicolumn{2}{c}{
			\infer[\rulename{{\To} \text{\textsf{at}}}]
			{\Gamma \To \Delta,\,\s:P(\bar x)}
			{\{\,\Gamma \To \Delta,\,\{k\}:P(\bar x)\mid k\in \s\,\}}
		}
		\\[2.8ex]
		
		\infer[\rulename{{\To}\land}]
		{\Gamma \To \Delta,\,\s:\phi\land\psi}
		{\Gamma \To \Delta,\,\s:\phi \qquad \Gamma \To \Delta,\,\s:\psi}
		&
		\infer[\rulename{\land{\To}}]
		{\s:\phi\land\psi,\,\Gamma \To \Delta}
		{\s:\phi,\,\s:\psi,\,\Gamma \To \Delta}
		\\[2.4ex]
		
		\infer[\rulename{{\To}\lori}]
		{\Gamma \To \Delta,\,\s:\phi\lori\psi}
		{\Gamma \To \Delta,\,\s:\phi,\,\s:\psi}
		&
		\infer[\rulename{\lori{\To}}]
		{\s:\phi\lori\psi,\,\Gamma \To \Delta}
		{\s:\phi,\,\Gamma \To \Delta \qquad \s:\psi,\,\Gamma \To \Delta}
		\\[2.8ex]
		
		\multicolumn{2}{c}{
			\infer[\rulename{{\To}{\to}}]
			{\Gamma \To \Delta,\,\s:\phi\to\psi}
			{\{\,\t:\phi,\,\Gamma \To \Delta,\,\t:\psi \mid \t\subseteq\s\,\}}
		}
		\\[2.8ex]
		
		\multicolumn{2}{c}{
			\infer[\rulename{{\to}{\To}}\;\text{\small where } \t\subseteq\s]
			{\s:\phi\to\psi,\,\Gamma \To \Delta}
			{\s:\phi\to\psi,\,\Gamma \To \Delta,\,\t:\phi
				\qquad
				\t:\psi,\,\s:\phi\to\psi,\,\Gamma \To \Delta}
		}
		\\[3.0ex]
		
		\infer[\rulename{{\To}\forall}^{\dagger}]
		{\Gamma \To \Delta,\,\s:\forall x\phi}
		{\Gamma \To \Delta,\,\s:\phi[z/x]}
		&
		\infer[\rulename{\forall{\To}}]
		{\s:\forall x\phi,\,\Gamma \To \Delta}
		{\s:\phi[y/x],\,\s:\forall x\phi,\,\Gamma \To \Delta}
		\\[2.4ex]
		%
		$\infer[\rulename{{\To}\Box}^{\ddagger}]{\Gamma\;\To\;\Delta,\s:\Box\phi}{\{\w\R\t,\Gamma\;\To\;\Delta, \t:\phi \mid \w\in\s\}}$
		&
		\infer[\rulename{\Box{\To}}\text{ where $\w\in\s$}]{\s:\Box\phi,\;\w\R\t,\;\Gamma\;\To\;\Delta}{\t:\phi,\;\s:\Box\phi,\;\w\R\t,\;\Gamma\;\To\;\Delta}
		\\[1.5ex]
		\multicolumn{2}{c}{
			\textnormal{\footnotesize $\dagger=$ \lq\lq{}$z$ does not occur in the conclusion" \hspace*{1 cm}$\ddagger=$ \lq\lq$\# t=n_\phi$ and $t\cap (W_{\Gamma,\Delta}\cup s)=\emptyset$"}
		}
	\end{array}
	$			
\end{figure}

\noindent		
As mentioned in the introduction, the rules $(\Box{\To})$ and $({\To}\Box)$ are inspired by standard rules for $\Box$ in labelled sequent calculi \cite{Negri:05}. The side condition for the right rule, however, is connected to the finite coherence property stated in Proposition \ref{Finite-coherence}. The rule can be read as follows: if for each world $w$ in a state $s$, an arbitrary set $t$ of successors 
of $w$ of cardinality $\le n_\phi$ supports $\phi$, then $s$ supports $\Box\phi$. This rule is sound (as we will see) since the support of $\phi$ on an arbitrary set $t$ of at most $n_\phi$-many successors of $w$ guarantees its support on the whole set $R[w]$ of successors by Proposition \ref{Finite-coherence}.
\begin{definition}[Derivability]\label{def:derivability}
	A derivation in \IWMC{} is a finite tree generated using the inference rules of \IWMC{} whose initial sequents are instances of $(\textsf{id})$ or $(\bot{\To})$. We say that a sequent $\Gamma\To\Delta$ is derivable in \IWMC{}, written $\vdash\Gamma\To\Delta$, if there is a derivation having $\Gamma\To\Delta$ as its root. 
\end{definition}

\noindent
It is clear by inspecting the rules that \IWMC{} satisfies an appropriate notion of analyticity. Let the set of subformulas of a formula $\phi$, $\textsf{SubF}(\phi)$ be given in the usual way. Then, if we define the set of subexpressions of $\phi$ as $\textsf{Sub}(\phi)=\textsf{SubF}(\phi)\cup\{\psi[y/x]\mid \psi\in\textsf{SubF}(\phi),\ y\text{ free for $x$ in $\psi$}\}$, the following holds.

\begin{proposition}[Analyticity]
Given a derivation of a sequent $\Gamma\To\Delta$ in \IWMC{}, for any labelled formula $\t:\psi$ appearing in the derivation there is some labelled formula $\s:\phi\in\Gamma\cup\Delta$ such that $\psi\in\textsf{Sub}(\phi)$.
\end{proposition}

\subsection{Interpretation and validity for sequents}

We start by defining a notion of satisfaction for labelled formulas, relational atoms and labelled sequents. For this, we make use of the notion of \textit{mappings} into a model, i.e. functions associating each natural number to some possible world in the universe of the model.
\begin{definition}[Satisfaction]
	Given a model $M=\langle W,R,D,I\rangle$, an assignment $g:Var\to D$ and a mapping $f:\N\to W$, we define satisfaction for labelled formulas $\s:\phi$ and relational atoms $\w\R\vv$ as follows, 
	\vspace*{-.1cm}
	\begin{align*}
		&M,f\Vdash_g \s:\phi\iff M,f[\s]\models_g\phi\\[3.5pt]
		&M,f\Vdash_g \w\R\vv\iff f(\w) R f(\vv)
	\end{align*}
	For multisets $\Gamma$ of labelled formulas or relational atoms, we write $M,f\Vdash_g \Gamma$ if $M,f\Vdash_g\gamma$ for all $\gamma\in\Gamma$. Note that in particular, for a set $\w\R\s=\{\w\R\vv\mid\vv\in\s\}$ of relational atoms we have
	$$M,f\Vdash_g\w\R\s\iff f(\w)R f(\vv)\text{ for all }\vv\in\s\iff f[\s]\subseteq R[f(\w)]$$
	Finally, for sequents $\Gamma\To\Delta$, we let: 
	$$\text{$M,f\Vdash_g \Gamma\To\Delta\iff M,f\Vdash_g\Gamma$ implies $M,f\Vdash_g\delta$ for some $\delta\in\Delta$}$$
\end{definition}

\noindent
We may now define a sequent to be \emph{valid} if it is satisfied under any interpretation.

\begin{definition}[Validity] We say that a sequent $\Gamma\To\Delta$ is valid, and write $\models\Gamma\To\Delta$, if for any model $M$, assignment $g$ and mapping $f$ into $M$ we have $M,f\Vdash_g \Gamma\To\Delta$.
\end{definition}

\noindent
The notion of validity for sequents is connected to the notion of entailment between \inqwqml-formulas via coherence. Indeed, due to Proposition \ref{Entailment-finite-coherence}, an entailment $\Phi\models\psi$ is valid iff it is valid on arbitrary states of cardinality at most $n_\psi$ (where $n_\psi$ is the coherence estimate for $\psi$ given by \ref{Finite-coherence}). These states are exactly the possible interpretations, via mappings, of a label $\s$ of cardinality $n_\psi$. This underpins the following connection.

\begin{proposition}\label{prop:connection} For any finite set of \inqwqml{} formulas $\Phi\cup\{\psi\}$ and an arbitrary label $\s$ with $\#\s\ge n_\psi$, letting $\s:\Phi$ denote $\{\s:\phi\mid \phi\in\Phi\}$ we have
$$\Phi\models\psi\quad\text{ iff }\quad \text{the sequent }(\s:\Phi\To \s:\psi)\text{ is valid}$$
\end{proposition}
\begin{proof} 
$(\implies):$ By contraposition. Assume that $\s:\Phi\To\s:\psi$ is not valid. Then, there is some model $M$, mapping $f$ and assignment $g$ such that $M,f\nVdash_g\s:\Phi\To\s:\psi$. Therefore, $M,f[\s]\Vdash_g\phi$ for all $\phi\in\Phi$ and $M,f[\s]\nVdash_g\psi$. It follows that $\Phi\nmodels\psi$.

%
%
%
%

$(\impliedby):$ By contraposition. Assume that $\Phi\nmodels\psi$. Then, by Proposition \ref{Entailment-finite-coherence}, $\Phi\nmodels_{n_\psi}\psi$, i.e., there exists a model $M$
, a state $s$ with $\# s\leq n_\psi$ and an assignment $g$ such that $M,s\models_g\phi$ for all $\phi\in\Phi$ and $M,s\nmodels_g\psi$. Now, let $f$ be any mapping such that $f[\s]=s$. Such a mapping exists, since $\# s\leq n_\psi\leq \#\s$. Clearly, $M,f\Vdash_g\s:\phi$ for all $\phi\in\Phi$ and $M,f\nVdash_g\s:\psi$, which implies that $M,f\nVdash_g\s:\Phi\To\s:\psi$.
\end{proof}

\noindent
Given this connection, we say that a derivation in our proof system is a \emph{proof} of an entailment $\Phi\models\psi$ in case its conclusion is a sequent of the form $\s:\phi_1,\dots,\s:\phi_n\To\s:\psi$ for some formulas $\phi_1,\dots,\phi_n\in\Phi$ and some label $\s$ with $\#\s\ge n_\psi$.


\subsection{Illustration}

\noindent
To illustrate the calculus, we provide a proof of  the \inqwqml{} entailment $\Box\forall x{?}Px\models 
	\Box{?}\forall xPx$. Since the coherence estimate for the conclusion is $n_{\Box{?}\forall xPx}=1$, it suffices to give a proof of a sequent of the form $\{\w\}:\Box\forall x{?Px}\To \{\w\}:\Box{?\forall x Px}$, involving a label of size 1. 
	
	To ease notation, we display labels as sequences of indices rather than sets; for instance, instead of $\{1\}:\phi$ we write simply $1:\phi$. For simplicity, we make use of weakening rules, whose admissibility will be proved in Section \ref{subsection-structural}, and we occasionally apply multiple rules in one step, when no ambiguity arises. We also derive two branches, $\bigstar_2$ and $\bigstar_3$, separately, due to space constraints. 

	\begin{prooftree}
		\AxiomC{}
		\RightLabel{(\textsf{id})}
		\UnaryInfC{$23:Py,\ 23:\forall xPx,\ 23:\forall x{?}Px\To 23:Py,\ 23:\bot$}
		\RightLabel{($\forall{\To}$)}
		\UnaryInfC{$23:\forall xPx,\ 23:\forall x{?}Px\To 23:Py,\ 23:\bot$}
		\AxiomC{\vdots}
		\noLine
		\UnaryInfC{($\bigstar_2$)}
		\AxiomC{\vdots}
		\noLine
		\UnaryInfC{($\bigstar_3$)}
		\RightLabel{(${\To\to}$)}
		\TrinaryInfC{$23:\forall x{?}Px\To 23:\lnot \forall xPx,\ 23:Py$}
		\RightLabel{(${\To}\forall$)}
		\UnaryInfC{$23:\forall x{?}Px\To 23:\forall xPx,\ 23:\lnot \forall xPx$}
		\RightLabel{(${\To}\lori$)}
		\UnaryInfC{$23:\forall x{?}Px\To 23:{{?}\forall xPx}$}
		\RightLabel{($\w{\To}$)}
		\UnaryInfC{$23:\forall x{?}Px,\ 1\R 23,\ 1:\Box\forall x{?}Px\To 23:{{?}\forall xPx}$}
		\RightLabel{($\Box{\To}$)}
		\UnaryInfC{$1\R2,\ 1\R3,\ 1:\Box\forall x{?}Px\To 23:{?}\forall xPx$}
		\RightLabel{(${\To}\Box$)}
		\UnaryInfC{$1:\Box\forall x{?}Px\To 1:\Box{?}\forall xPx$}
	\end{prooftree}
	\hspace*{0cm}\\
The steps of this first half of the derivation are straightforward, and no choices need to be made to perform them. The only exception is the final application of ($\forall{\To}$) in the leftmost branch, where using the variable $y$ is the only reasonable choice to obtain an ($\textsf{id}$) initial sequent. The application of rule $({\To\to})$ generates three branches, one for each  non-empty subset of the label $23$. The branch corresponding to the subset $23$ is displayed in the proof; the branch $\bigstar_2$ corresponding to the subset $2$ is shown below; finally, the branch $\bigstar_3$ is analogous to $\bigstar_2$, but with the label 3 playing the role of the label 2.

	\begin{prooftree}
	\AxiomC{}
	\RightLabel{(\textsf{id})}
	\UnaryInfC{$23:Py, \dots\To 23:Py$}
	\AxiomC{}
	\RightLabel{(\textsf{id})}
	\UnaryInfC{$\dots,\ 2:Py\To \dots,\ 2:Py$}
	\AxiomC{}
	\RightLabel{($\bot{\To}$)}
	\UnaryInfC{$2:\bot,\dots\To \dots$}
	\RightLabel{(${\to\To}$)}
	\BinaryInfC{$23:\lnot Py,\ 2:Py\To 23:Py$}
	\RightLabel{($\lori{\To}$)}
	\BinaryInfC{$2:Py,\ 23:{{?}Py}\To 23:Py$}
	\RightLabel{($\w{\To}$),(${\To}\w$)}
	\UnaryInfC{$2:Py,\ 23:{?}Py,\ 2:\forall xPx,\ 23:\forall x{?}Px\To 23:Py,\ 2:\bot$}
	\RightLabel{($\forall{\To}$),($\forall{\To}$)}
	\UnaryInfC{$2:\forall xPx,\ 23:\forall x{?}Px\To 23:Py,\ 2:\bot$}
	\noLine
	\UnaryInfC{$\underbrace{\hspace*{.375\textwidth}
		}_{\bigstar_2}$}
\end{prooftree}\hspace*{0cm}\\
Semantically, providing a derivation of branch $\bigstar_2$ corresponds to showing that, if the truth value of $Py$ is uniform in the state $23$ (which we know from $23:\forall x{?}Px$) and the world $2$ makes $Py$ true (from $2:\forall x Px$), then we can show that $Py$ is true across the whole state $23$. Therefore, we proceed as follows: we apply ($\forall{\To}$) twice and instantiate it using the variable $y$. Then, in the right branch generated by ($\lori{\To}$), we apply (${\to}{\To}$) picking $2$ as our choice of subset for the label $23$. This allows us to use the information that $Py$ is true at world $2$ to conclude the derivation.

\subsection{Soundness.} 

\noindent
The soundness of our system amounts to the claim that if a sequent is derivable, it is valid.

\begin{proposition}[Soundness]\label{soundness} For any sequent $\Gamma\To\Delta$ we have: 
 $\vdash \Gamma\To\Delta$ implies $\models \Gamma\To\Delta$.
 
\end{proposition}
\begin{proof}
	We prove soundness by induction on the structure of a derivation of $\Gamma\To\Delta$. Therefore, it suffices to show that validity is preserved by the rules. Below, we give the proof steps for the cases of ($\Box{\To}$) and (${\To}\Box$). The proof of the remaining cases remains unchanged from \cite{LitakSano:25}.
	
	\medskip\noindent
	\underline{$(\Box{\To})$:} Assume that, for some $\w\in \s$, the sequent $(\Gamma,\s:\Box\phi,\w\R\t,\t:\phi\To\Delta)$ is satisfied by all models, mappings and assignments. Let $M=\langle W,R,D,I\rangle$ be a model, $g$ an assignment and $f$ a mapping over $M$. Now, suppose that $M,f\Vdash_g \Gamma,\s:\Box\phi,\w\R\t$. Then, in particular, we have that $M,f[\s]\models_g\Box\phi$ and $f[\t]\subseteq R[f(\w)]$. By the semantics of $\Box$ and by persistency, this gives us $M,f[\t]\models_g\phi$, or, equivalently, $M,f\Vdash_g \t:\phi$. Therefore, we have $M,f\Vdash_g\Gamma,s:\Box\phi,\w\R\t,\t:\phi$. By our initial assumption, this implies that $M,f\Vdash_g\delta$ for some $\delta\in\Delta$. This shows that the sequent  $(\Gamma,\s:\Box\phi,\w\R\t\To\Delta)$ is valid.
	
	\medskip\noindent
	\underline{$({\To}\Box)$:} Assume that there is a label $\t$ with $\# \t= n_\phi$, $\t\cap (W_{\Gamma,\Delta}\cup\s)=\emptyset$ and such that, for all $\w\in\s$, for any model, mapping and assignment, the sequent $(\Gamma,\w\R\t\To\Delta,\t:\phi)$ is satisfied. Then, let $M=\langle W,R,D,I\rangle$ be a model, $g$ an assignment and $f$ a mapping over $M$, and assume that $M,f\Vdash_g \Gamma$, and that for all $\delta\in\Delta$, $M,f\nVdash_g\delta$. We prove by contradiction that $M,f\Vdash_g \s:\Box\phi$. Suppose otherwise. We have:
	\begin{align*}
		&\text{$M,f\nVdash_g\s:\Box\phi$},\\
		\text{$\iff$ }&\text{$M,f[\s]\nmodels_g\Box\phi$},\\
		\text{$\iff$ }&\text{for some $w\in f[\s]$, $M,R[w]\nmodels_g\phi$}\\
		\text{$\iff$ }&\text{for some $\w\in \s$, $M,R[f(\w)]\nmodels_g\phi$}\\
		\text{$\iff$ }&\text{for some $\w\in \s$ and some $s'\subseteq R[f(\w)]$ s.t. $\#s'\leq n_\phi:\ M,s'\nmodels_g\phi$.} & \text{(by Prop.\ \ref{Finite-coherence})}
	\end{align*} 
	Now, we define a new mapping $f'$ such that $f'[\t]=s'$ and $f'$ coincides with $f$ on all indices not in \t. Note that we can define such an $f'$ because $\# \t=n_\phi\geq \# s'$. For some $\w\in\s$ we have $M,f'\Vdash_g \w\R\t$ and $M,f'\nVdash_g\t:\phi$, since $f'[\t]=s'$, $s'\subseteq R[f(\w)]$ and $M,s'\nmodels_g\phi$. Moreover, $M,f'\Vdash_g\Gamma$ and $M,f'\nVdash_g\delta$ for all $\delta\in\Delta$, since $\t\cap (W_{\Gamma,\Delta}\cup s)=\emptyset$, so $f'$ agrees with $f$ on all indices occurring in $\Gamma$ or $\Delta$.
	Therefore, for some $\w\in\s$, the sequent $(\Gamma,\w\R\t\To\Delta,\t:\phi)$ is not satisfied by $M,f',g$, contradicting the assumption that this sequent is valid.
\end{proof}

\noindent
We have thus proved that our proof system \IWMC{} is sound for \inqwqml{}. Before proving that it is also complete, we study the proof-theoretic properties of the calculus.

\subsection{Structural properties}\label{subsection-structural}
In this section, we show the admissibility of weakening, contraction and cut rules for \IWMC{}, as well as the invertibility of all rules in the system. Besides their intrinsic interest, these results will play a role below in the completeness proof.

We will use the notion of \textit{height} of a derivation, defined as the length of one of its longest branches, and we will write $\vdash_n\Gamma\To\Delta$ when there is a derivation of $\Gamma\To\Delta$ of height at most $n$. We will say that a rule is \textit{admissible} if whenever its premisses are derivable, then its conclusion is also derivable. A rule is said to be \textit{height-preserving admissible} if whenever all its premisses have a derivation of height at most $n$, then its conclusion also has a derivation of height at most $n$. By a rule being \textit{invertible}, we mean that all rules having its conclusion as premiss and one of its premisses as conclusion are admissible.  
\begin{definition}[Label substitution]
Given a multiset $\Gamma$ of labelled formulas and relational atoms  and indices $\w,\vv\in\N$, we write $\Gamma[\vv/\w]$ to indicate the multiset of labelled formulas obtained by replacing, in each element of $\Gamma$, every occurrence of $\w$ with an occurrence of $\vv$. Similarly, given a label $\s=\{\w_1,\dots,\w_n\}$ with $\w_1<\w_2<\dots<\w_n$, and a sequence $\u=\vv_1,\dots,\vv_n$ of (not necessarily distinct) indices, whose length matches the size of $\s$,
we write $\Gamma[\u/\s]$ to mean $\Gamma[\vv_1/\w_1,\dots,\vv_n/\w_n]$, where all substitutions are performed simultaneously.
\end{definition}
\begin{lemma}[Relabeling lemma]\label{relabeling-lemma}
If $\vdash_n\Gamma\To\Delta$, then $\vdash_n\Gamma[\u/\s]\To\Delta[\u/\s]$ for any label $\s$ and sequence of natural numbers $\u$ s.t. $\textsf{length}(\u)=\# \s$. 
\end{lemma}
\begin{proof}
By induction on the height $n$ of the derivation. If $n=0$, the sequent is either an instance of $(\textsf{id})$ or $(\bot{\To})$ and the result is trivial. Assume then that the claim holds for all $k< n+1$. Then we can proceed by cases on the last rule applied. The cases of $\textsf{at},\land,\lori,\to,\forall$-rules all follow immediately by induction (note that for $(\to{\To})$, the condition on the labels remains satisfied after any substitution), and so does the $(\Box{\To})$ case. This leaves out only the $({\To}\Box)$ case.

Assume $\vdash_{n+1}\Gamma\;\To\;\Delta,\s:\Box\phi$ and that for some $\t$ s.t. $\# \t= n_\phi$ and $\t\cap (W_{\Gamma,\Delta}\cup \s)=\emptyset$: for all $\w\in \s$, $\vdash_n \Gamma, \w\R\t\;\To\;\Delta, \t:\phi$. 

If $\u\cap \t=\emptyset$, we apply the inductive hypothesis to get that $\vdash_n\Gamma[\u/\s], \w\R\t[\u/\s]\;\To\;\Delta[\u/\s], \t:\phi[\u/\s]$ for all $\w\in \s$. We conclude by applying $({\To}\Box)$, which is possible because $\t\cap\s=\emptyset$ and, therefore, $\t[\u/\s]=\t$. As $\u\cap\t=\emptyset$, $\t[\u/\s]$ still satisfies the side conditions. 

If $\u\cap\t\neq \emptyset$, we first apply the inductive hypothesis to replace $\t$ with a label $\t'$ of equal cardinality, but disjoint from both $\u$, $\t$ and $W_{\Gamma,\Delta}\cup\s$, and we get: for all $\w\in\s$, $\vdash_n \Gamma, \w\R\t'\;\To\;\Delta, \t':\phi$. Since $\t'$ satisfies the side conditions of $({\To}\Box)$, we then apply the inductive hypothesis to perform the label substitution $[\u/\s]$, and we conclude as in the previous case.
\end{proof}
\begin{proposition}\label{prop-structural-properties}
Let $\sigma$ stand for either a labelled formula $\u:\phi$ or a relational atom $\w\R\vv$.
\begin{enumerate}
	\item The following weakening rules are height-preserving admissible in \IWMC{}:
	\begin{center}$\begin{array}{cc}
			\infer[\rulename{{\To}{\w}}]
			{\Gamma \To \Delta,\u:\phi}
			{\Gamma \To \Delta}
			&
			\infer[\rulename{{\w}{\To}}]
			{\Gamma, \sigma \To \Delta}
			{\Gamma \To \Delta}
		\end{array}$
	\end{center}
	\item All rules of \IWMC{} are height-preserving invertible.
	\item The following contraction rules are height-preserving admissible in \IWMC{}:
	\begin{center}$\begin{array}{cc}
			\infer[\rulename{{\To}{\textsf{c}}}]
			{\Gamma \To \Delta,\u:\phi}
			{\Gamma \To \Delta,\u:\phi,\u:\phi}
			&
			\infer[\rulename{{\textsf{c}}{\To}}]
			{\Gamma, \sigma \To \Delta}
			{\Gamma,\sigma,\sigma\To \Delta}
		\end{array}$
	\end{center}
\end{enumerate}
\end{proposition}
\begin{proof}[Proof of $1$]
By induction on the height of the derivation of the premiss $\Gamma\To\Delta$. The case of $n=0$ is standard for both weakening rules. For the inductive case, assume $\vdash_{n+1}\Gamma\To\Delta$. The proof proceeds by cases on the last applied rule in the derivation. We only detail the cases involving the modal rules for $({\To}\w)$: the modal cases for $(\w{\To})$ are analogous, and all non-modal cases are standard.
\begin{itemize}
	\item\underline{$(\Box{\To}):$} If the last applied rule in the derivation is $(\Box{\To})$, then $\Gamma=\Gamma',\s:\Box\psi,\w\R\t$ for some $\w\in \s$ and the premiss is $\Gamma',\s:\Box\psi,\w\R\t,\t:\psi\To\Delta$. Therefore, $\vdash_n \Gamma',\s:\Box\psi,\w\R\t,\t:\psi\To\Delta$. By induction, $\vdash_n\Gamma',\s:\Box\psi,\w\R\t,\t:\psi\To\Delta,\u:\phi$. By applying $(\Box{\To})$, we then find a derivation of height $n+1$ of $\Gamma\To\Delta,\u:\phi$.
	\item\underline{$({\To}\Box):$} If the last applied rule in the derivation is $({\To}\Box)$, then $\Delta=\Delta',\s:\Box\psi$, and, for some $\t $ s.t. $\# \t= n_\psi$ and $\t\cap W_{\Gamma,\Delta}=\emptyset$: for all $\w\in\s$, $\vdash_n\Gamma,\w\R\t\To\Delta',\t:\psi$. Now, we distinguish two cases. If $\t\cap \u=\emptyset$, we proceed as above. If $\t\cap \u\neq\emptyset$, thanks to Lemma \ref{relabeling-lemma} and to the fact that $\u$ and $W_{\Gamma,\Delta}$ are finite, we can replace the label $\t$ with a label $\t'$ s.t. $\t'\cap (\t\cup W_{\Gamma,\Delta}\cup \u)=\emptyset$ and we have, for all $\w\in\s$, $\vdash_n\Gamma,\w\R\t'\To\Delta',\t':\psi$. By the inductive hypothesis, for all $\w\in\s$, $\vdash_n\Gamma,\w\R\t'\To\Delta',\u:\phi,\t':\psi$. Since $\t'$ still satisfies the side conditions of $({\To}\Box)$, we have $\vdash_{n+1}\Gamma\To\Delta,\u:\phi$.\qedhere
\end{itemize}
\end{proof}
\begin{proof}[Proof of $2$]
To show that all rules are invertible, we structure our proof following standard arguments from \cite{Negri:01}. We proceed one rule at a time, by induction on the height of the derivation. The base step of the induction is straightforward and standard for both the first-order rules and the modal rules. The inductive step is divided into four cases, depending on (i) whether the last rule applied is the one being inverted and (ii) whether the principal formula of the last applied rule is principal in the inverted rule. Below, we only detail the steps involving the modal rules, as the rest remain unchanged from \cite{LitakSano:25}. 

\medskip\noindent
\underline{Invertibility of the non-modal rules.} We need to add two clauses to the case where the last rule applied is not the one being inverted, and it is a modal rule. If the last rule is $(\Box{\To})$, the proof is straightforward. If the last rule is $({\To}\Box)$, the argument is the same, but it requires observing that the natural numbers occurring in labels in the premisses are always a subset of those found in the conclusion, so the side conditions on $\t$ remain satisfied. 

\medskip\noindent
\underline{Invertibility of $(\Box{\To})$}. Follows immediately from the height-preserving admissibility of weakening.

\medskip\noindent
\underline{Invertibility of $({\To}\Box)$.} Assume $\vdash_{n+1}\Gamma\To\Delta,\s:\Box\phi$. We distinguish several cases:
\begin{itemize}
	\item Last applied rule is $({\To}\Box)$ and $\s:\Box\phi$ is principal: immediate.
	\item Last applied rule is $({\To}\Box)$ and $\s:\Box\phi$ is not principal: then, for some $\u$ and $\psi$, letting $\Delta=\Delta',\u:\Box\psi$ we have $\vdash_{n+1}\Gamma\To\Delta',\u:\Box\psi,\s:\Box\phi$ and, for some $\t$ s.t. $\#\t=n_\psi$ and $\t\cap(W_{\Gamma,\Delta'}\cup\s\cup\u)=\emptyset$: for all $\vv\in\u$, $\vdash_{n}\Gamma,\vv\R\t\To\Delta',\t:\psi,\s:\Box\phi$. 
		
		By inductive hypothesis, for each $\vv\in\u$ there is $\t_\vv$ with $\#\t_\vv=n_\phi$ and $\t_\vv\cap(W_{\Gamma,\Delta'}\cup\s\cup\t\cup\{\vv\})=\emptyset$ and such that: for all $\w\in\s$, $\vdash_{n}\Gamma,\vv\R\t,\w\R\t_\vv\To\Delta',\t:\psi,\t_\vv:\phi$. We then apply Lemma \ref{relabeling-lemma} to substitute 
		all the labels $\t_\vv$ with a single label $\t_\s$, taken to be disjoint from any labels included so far. As a result, we can reframe the set of derivable sequents as follows. For any $\w\in\s$, for all $\vv\in\u$, $\vdash_{n}\Gamma,\vv\R\t,\w\R\t_\s\To\Delta',\t:\psi,\t_\s:\phi$. For each $\w\in\s$, we can consider the corresponding set of sequents, to which, since $\t$ still satisfies the side conditions, we can then apply rule $({\To}\Box)$ (with $\vv\R\t$ and $\t:\psi$ principal), and conclude $\vdash_{n+1}\Gamma,\w\R\t_\s\To\Delta',\u:\Box\psi,\t_\s:\phi$ for all $\w\in\s$.
	\item Last applied rule is not $({\To}\Box)$: we only show how to proceed for the rule $({\To}\forall)$, which has side conditions, and for the case of $({\To} \textsf{at})$, which requires additional observations. The case of $({\To}\to)$ is analogous to that of $({\To}\textsf{at})$.
	\begin{itemize}
		\item Last applied rule is $({\To}\forall)$: then, letting $\Delta=\Delta',\u:\forall x\psi$, we have $\vdash_{n+1}\Gamma\To\Delta',\s:\Box\phi,\u:\forall x\psi$ and, for some $z$ not occurring in the conclusion, $\vdash_{n}\Gamma\To\Delta',\s:\Box\phi,\u:\psi[z/x]$. By inductive hypothesis, we get, for some $\t$ s.t. $\# \t= n_\phi$ and $\t\cap(W_{\Gamma,\Delta'}\cup\s\cup \u)=\emptyset$: for all $\w\in\s$, $\vdash_{n}\Gamma,\w\R\t\To\Delta',\t:\phi,\u:\psi[z/x]$. We apply $({\To}\forall)$ (since $z$ did not occur in the conclusion and, therefore, in $\phi$), and conclude that, for some $\t$ s.t. $\# \t= n_\phi$ and $\t\cap(W_{\Gamma,\Delta'}\cup\s\cup \u)=\emptyset$: for all $\w\in\s$, $\vdash_{n+1}\Gamma,\w\R\t\To\Delta',\t:\phi,\u:\forall x\psi$.
		\item Last applied rule is $({\To}\textsf{at})$: then,  letting $\Delta=\Delta',\u:P(\bar x)$, we have $\vdash_{n+1}\Gamma\To\Delta',\s:\Box\phi,\u:P(\bar x)$ and, for all $k\in \u$, $\vdash_{n}\Gamma\To\Delta',\s:\Box\phi,k:P(\bar x)$. Then, by inductive hypothesis, we get that, for all $k\in \u$, there is some $\t_k$ such that $\# \t_k= n_\phi$, $\t_k\cap(W_{\Gamma,\Delta'}\cup\s\cup\{k\})=\emptyset$ and such that for all $\w\in\s$, $\vdash_{n}\Gamma,\w\R\t_k\To\Delta',\t_k:\phi,k:P(\bar x)$. Then, for each $k\in \u$, we apply Lemma \ref{relabeling-lemma} to each sequent in the corresponding set and use it to replace $\t_k$ with $\t'$, a label of the same cardinality which is disjoint both from $W_{\Gamma,\Delta}$, $\s$, and $\u$ and from $\bigcup_{k\in \u} \t_k$ (such $\t'$ exists, since $\Gamma$ and $\Delta$ are finite). Then, we have that for some $\t'$ s.t. $\# \t'= n_\phi$ and $\t'\cap(W_{\Gamma,\Delta}\cup\s\cup \u)=\emptyset$ 
		for all $k\in \u$, for all $\w\in\s$, $\vdash_{n}\Gamma,\w\R\t'\To\Delta',t':\phi,k:P(\bar x)$. We then apply $({\To}\textsf{at})$ to get that for some $\t'$ s.t. $\# \t'= n_\phi$ and $\t'\cap(W_{\Gamma,\Delta}\cup\s\cup \u)=\emptyset$: for all $\w\in\s$, $\vdash_{n+1}\Gamma,\w\R\t'\To\Delta',t':\phi,\u:P(\bar x)$\qedhere
	\end{itemize} 
\end{itemize}
\end{proof}	 

\begin{proof}[Proof of $3$] By induction on the height of the derivation of the premiss. The initial step of the induction is straightforward.
The inductive step is divided by cases, depending on the last rule applied in the derivation and on whether the contracted formula was principal in the last step. Those cases where the contracted formula is not principal in the last applied rule in the derivation are standard: we apply the inductive hypothesis to each premiss, then apply the last rule to the contracted sequents.

Of the cases where the formula being contracted is principal in the last applied rule, we only spell out those that were not subsumed in the first-order case:

\medskip\noindent
\underline{$({\Box}{\To})$ for $(\textsf{c}{\To})$:}  
There are two possibilities: either we want to contract a double occurrence of $s:\Box\phi$ or one of $\w\R\vv$. Both follow easily, by applying the inductive hypothesis to the premiss, where the duplicate formula is still present.

\medskip\noindent
\underline{$({\To}{\Box})$ for $({\To}\textsf{c})$:}	Assuming contraction holds up to height $n$, suppose that $\vdash_{n+1}\Gamma\To\Delta,\s:\Box\phi,\s:\Box\phi$. Then, for some $\t$ s.t. $\# \t= n_\phi$ and $\t\cap(W_{\Gamma,\Delta}\cup \s)=\emptyset$: for all $\w\in \s$, $\vdash_{n}\Gamma,\w\R\t\To\Delta,\s:\Box\phi,\t:\phi$. By the height-preserving invertibility of $({\To}\Box)$, we get that for the same $\t$ as above, for all $\w\in \s$, there is a $\t_w$ with $\#\t_w=n_\phi$ and $\t_w\cap(W_{\Gamma,\Delta}\cup \s\cup\t)=\emptyset$ and such that: for all $\vv\in\s$, $\vdash_{n}\Gamma,\w\R\t,\vv\R\t_w\To\Delta,\t:\phi,\t_w:\phi$. By Lemma \ref{relabeling-lemma}, since both $\t$ and $\t_\w$ have cardinality $n_\phi$, we can perform the label substitution $[\t/\t_\w]$ and get that for all $\w,\vv\in\s$, $\vdash_{n}\Gamma,\w\R\t,\vv\R\t\To\Delta,\t:\phi,\t:\phi$. In particular, this holds in the case where $\vv=\w$, so we have that, for all $\w\in\s$, $\vdash_{n}\Gamma,\w\R\t,\w\R\t\To\Delta,\t:\phi,\t:\phi$. By inductive hypothesis, we can apply left contraction $\#\t$ times (recall that $\w\R\t=\{\w\R\vv\mid\vv\in\t\}$) and right contraction once to get $\vdash_{n}\Gamma,\w\R\t\To\Delta,\t:\phi$ for all $\w\in\s$. The conclusion then follows by applying $({\To}\Box)$.
\end{proof}
\begin{proposition}\label{cut}
The rule of cut is admissible in \IWMC{}:
\begin{center}
	\begin{prooftree}
		\AxiomC{$\Gamma \To \Delta,\s:\phi$}
		\AxiomC{$\s:\phi,\Pi\To\Sigma$}
		\RightLabel{\textnormal{(\textsf{cut})}}
		\BinaryInfC{$\Gamma,\Pi\To\Delta,\Sigma$}
	\end{prooftree}
\end{center}
\end{proposition}
\begin{proof}
	The proof follows a standard structure for proofs of cut admissibility: an induction on the length of the \textit{cut formula} $\phi$ and a subinduction on the sum of the heights of the derivations of the premisses. We start from two assumptions, $\vdash_n\Gamma\To\Delta,\s:\phi$ and $\vdash_m\s:\phi,\Pi\To\Sigma$, and we aim to reach the conclusion $\vdash\Gamma,\Pi\To\Delta,\Sigma$. Let us call a formula \textit{principal} in a premiss when it is principal in the last applied rule in the derivation of that premiss. There are two cases that we need to consider that cannot be transferred or straightforwardly adapted from the cut admissibility proof given in \cite{LitakSano:25}.
	
	First, suppose that the cut formula is not principal in at least one premiss and that the last applied rule in that premiss is $(\Box{\To})$ or $({\To}\Box)$. We spell out a representative case, as the others follow similar arguments. Assume that the cut formula $\s:\phi$ is not principal in the left premiss and that the last applied rule in that premiss is $({\To}\Box)$. Then, the assumption for the left premiss has the form $\vdash_n\Gamma\To\Delta',\u:\Box\psi,\s:\phi$, and it is the case that for some $\t$ s.t. $\#\t=n_\psi$ and $\t\cap(W_{\Gamma,\Delta'}\cup \s\cup\u)=\emptyset$, for all $\w\in\u$, $\vdash_{n-1}\Gamma,\w\R\t\To\Delta',\t:\psi,\s:\phi$. We can then combine the $(n-1)$-derivability of these sequents with the hypothesis about the right premiss, $\vdash_m\s:\phi,\Pi\To\Sigma$, to apply the cut rule (available here by subinduction hypothesis on the sum of the heights $n+m$) and obtain, for all $\w\in\u$, $\vdash\Gamma,\w\R\t,\Pi\To\Delta',\t:\psi,\Sigma$. At this point, if needed, we can invoke Lemma \ref{relabeling-lemma} to substitute $\t$ with a label $\t'$ such that $\t'\cap(W_{\Gamma,\Pi,\Delta',\Sigma}\cup\u\cup\t)=\emptyset$, and we can then conclude by an application of $({\To}\Box)$:
	\begin{prooftree}
		\AxiomC{$\{\Gamma,\w\R\t',\Pi\To\Delta',\t':\psi,\Sigma\mid\w\in\u\}$}
		\RightLabel{(${\To}\Box$)}
		\UnaryInfC{$\Gamma,\Pi\To\Delta',\u:\Box\psi,\Sigma$}
	\end{prooftree}
	
	Now, suppose instead that the cut formula is $\s:\Box\phi$ and that the last step of both derivations is a modal rule with $\s:\Box\phi$ principal. Then, from the right premiss, we get $\Pi=\Pi',\w\R\t$ for some $\w\in\s$ and $\vdash_{m-1} \Pi',\s:\Box\phi,\w\R\t,\t:\phi\To\Sigma$. From the left premiss, we have that for some $\u$ s.t. $\#\u=n_\phi$ and $\u\cap(W_{\Gamma,\Delta}\cup\s)=\emptyset$, for all $\vv\in\s$,  $\vdash_{n-1}\Gamma,\vv\R\u\To\Delta,\u:\phi$. From the last item, we can isolate the specific case of $\vv=\w$ and, by Lemma \ref{relabeling-lemma}, we can perform the substitution of $\u$ with $\t$, obtaining $\vdash_{n-1}\Gamma,\w\R\t\To\Delta,\t:\phi$. We prove the conclusion $\vdash \Gamma,\Pi',\w\R\t\To\Delta,\Sigma$ by constructing the following derivation:
	\begin{prooftree}
		\AxiomC{$\Gamma,\w\R\t\To\Delta,\t:\phi$}
		\AxiomC{$\Gamma\To\Delta,\s:\Box\phi$}
		\AxiomC{$\Pi',\s:\Box\phi,\w\R\t,\t:\phi\To\Sigma$}
		\RightLabel{(\textsf{cut})}
		\BinaryInfC{$\Gamma,\Pi',\w\R\t,\t:\phi\To\Delta,\Sigma$}
		\RightLabel{(\textsf{cut})}
		\BinaryInfC{$\Gamma,\w\R\t,\Gamma,\Pi',\w\R\t\To\Delta,\Delta,\Sigma$}
		\RightLabel{(\textsf{c}${\To}$),(${\To}$\textsf{c})}
		\UnaryInfC{$\Gamma,\Pi',\w\R\t\To\Delta,\Sigma$}
	\end{prooftree}
The first application of the cut rule is enabled by the subinductive hypothesis abouth the heights $n$ and $m$, while the second application relies on the inductive hypothesis about the length of $\phi$. The third step implicitly includes multiple applications of the contraction rules.
\end{proof}
\subsection{Completeness}

We prove completeness by showing that from any non-derivable sequent we can build a canonical model that refutes it. The proof generalizes the approach of \cite{LitakSano:25} to the modal case. First, we define a notion of saturation and prove that non-derivable sequents can be extended to saturated ones. Then, we show how to obtain a canonical model from a saturated sequent, and prove that this model refutes the original sequent. This establishes that each valid sequent is derivable. Finally, we lift this to a strong completeness theorem for \inqwqml{} by using the compactness and coherence properties of the logic.

For the proof, we need to extend our range of syntactic objects to include infinite sequents. Therefore, we start by defining a generalized class of sequents which we will call g-sequents. A g-sequent is an ordered pair $\Gamma\To\Delta$ consisting of (possibly infinite) multisets $\Gamma$ and $\Delta$, where $\Gamma$ can contain labelled formulas and relational atoms, while $\Delta$ contains only labelled formulas. For a g-sequent $\Gamma\To\Delta$, we define derivability in \IWMC{}, denoted by $\vdash\Gamma\To\Delta$, as the existence of a finite sequent $\Gamma'\To\Delta'$ with $\Gamma'\subseteq\Gamma$, and $\Delta'\subseteq\Delta$ which is derivable in the sense of Definition \ref{def:derivability}.
Note that every sequent $\Gamma\To\Delta$ is also a $g$-sequent, and that the derivability of $\Gamma\To\Delta$ as a sequent, as given by Definition \ref{def:derivability}, coincides with its derivability as a $g$-sequent by the admissibility of weakening. 
\paragraph{Saturation.} We say that a g-sequent $\Gamma\To\Delta$ is \emph{saturated} if it satisfies the following conditions:
\begin{itemize}[align=parleft,itemsep=0cm,labelwidth=.9cm,leftmargin=4em]
\item[]\hspace*{-1.075cm}$(\textsf{unprov})$ 
$\not\vdash\Gamma\To\Delta$
\item[$(\textsf{at}L)$] If $\s:P(x_1,\dots,x_n)\in\Gamma$, then $\{\w\}:P(x_1,\dots,x_n)\in\Gamma$ for all $\w\in\s$.
\item[$(\textsf{at}R)$] If $\s:P(x_1,\dots,x_n)\in\Delta$, then $\{\w\}:P(x_1,\dots,x_n)\in\Delta$ for some $\w\in\s$.
\item[$(\lori R)$] If $\s:\phi\lori\psi\in\Delta$, then $\s:\phi\in\Delta$ and $\s:\psi\in\Delta$.
\item[$(\lori L)$] If $\s:\phi\lori\psi\in\Gamma$, then $\s:\phi\in\Gamma$ or $\s:\psi\in\Gamma$.
\item[$(\Box R)$] If $\s:\Box\phi\in\Delta$, then for some $\w\in\s$ there is a $\t\subseteq\W_{\Gamma,\Delta}$ s.t. $(1)$ $\#\t= n_\phi$ $(2)$ $\w\R\t\subseteq\Gamma$ $(3)$ $\t:\phi\in\Delta$ 
\item[$(\Box L)$] If $\s:\Box\phi\in\Gamma$, then for all $\w\in\s$ and all finite $\t$ s.t. $\w\R\t\subseteq\Gamma$, $\t:\phi\in\Gamma$
\item[$(\land R)$] If $\s:\phi\land\psi\in\Delta$, then $\s:\phi\in\Delta$ or $\s:\psi\in\Delta$
\item[$(\land L)$] If $\s:\phi\land\psi\in\Gamma$, then $\s:\phi\in\Gamma$ and $\s:\psi\in\Gamma$
\item[$(\!\to\!\! R)$] If $\s:\phi\to\psi\in\Delta$, then for some $\t\subseteq\s$ $\t:\phi\in\Gamma$ and $\t:\psi\in\Delta$
\item[$(\!\to\!\! L)$]  If $s:\phi\to\psi\in\Gamma$, then for all $\t\subseteq\s$, $\t:\phi\in\Delta$ or $\t:\psi\in\Gamma$
\item[$(\forall R)$] If $\s:\forall x\phi\in\Delta$, then for some $z\in Var$, $\s:\phi[z/x]\in\Delta$
\item[$(\forall L)$] If $\s:\forall x\phi\in\Gamma$, then for all $z\in Var$, $\s:\phi[z/x]\in\Gamma$
\end{itemize}
\begin{lemma}[Saturation Lemma]\label{Saturation-Lemma} If a (finite) sequent $\Gamma\To\Delta$ is not derivable, there is (within a suitably extended language) a saturated g-sequent $\Gamma^+\To\Delta^+$ such that $\Gamma\subseteq\Gamma^+$ and $\Delta\subseteq\Delta^+$. We say that $\Gamma^+\To\Delta^+$ is a \emph{saturated extension} of $\Gamma\To\Delta$. 
\end{lemma}
\begin{proof} Take a non-derivable finite sequent $\Gamma\To\Delta$. We expand the language with countably many fresh variables $y_0,y_1,\dots$. Fix an enumeration $(\s_i:\phi_i)_{i\in\N}$ of all labelled formulas in the extended language such that each labelled formula occurs infinitely often in the enumeration; note that this is possible since labels are finite subsets of $\N$, and so they are countably many and, as a consequence, labelled formulas are countably many as well. Then for each $i\in\N$ we define inductively a family of finite multisets $\Gamma=\Gamma_0\subseteq\Gamma_1\subseteq\dots$ and  $\Delta=\Delta_0\subseteq\Delta_1\subseteq\dots$ such that for each $i$, $\not\vdash\Gamma_i\To\Delta_i$. For the inductive step, we distinguish a number of cases depending on the labelled formula $(\s_i:\phi_i)$. We spell out the details for the case in which $\phi_i$ is an atomic formula $P(x_1,\dots,x_n)$ and $s_i:\phi_i\in\Gamma_i$, and for the case in which $\phi_i$ is a modal formula $\Box\psi$, referring to \cite{LitakSano:25} for the remaining cases.

If $\phi_i$ is an atomic formula $P(x_1,\dots,x_n)$ and $s_i:\phi_i\in\Gamma_i$, we let $\Gamma_{i+1}=\Gamma_i\cup\{\w:\phi_i\mid\w\in\s_i\}$ and $\Delta_{i+1}=\Delta$. To show that the sequent $\Gamma_{i+1}\To\Delta_{i+1}$ is not derivable, we use the following notation: let $\bar{x}$ stand for $x_1,\dots,x_n$, let $\Gamma_i=\Gamma'_i,\s_i:P(\bar{x})$, let $\s=\{w_1,\dots,w_m\}$ and let $ID(k)$ stand for the sequent $\Gamma'_i,\s_i:P(\bar{x})\To\w_k:P(\bar{x})$. 
Note that $ID(k)$ is an (\textsf{id})-initial sequent, so it is derivable for any $k=1,\dots,m$. 
Then, we assume by contraposition the derivability of $\Gamma_{i+1}\To\Delta_{i+1}$, which we write in the extended form $\Gamma'_i,s:P(\bar{x}),\w_i:P(\bar{x}),\dots,\w_m:P(\bar{x})\To\Delta$, and show the derivability of $\Gamma_i\To\Delta_i$. The construction relies on the cut and contraction rules, and is built by iterating the derivation shown below to eliminate each $\w_k:P(\bar{x})$ using $ID(k)$. After $m$ iterations, we are left with $\Gamma_i\To\Delta_i$. 
	\begin{prooftree}
		\AxiomC{$ID(2)$}
		\AxiomC{($ID(1)$)}
		\noLine
		\UnaryInfC{$\Gamma'_i,\ \s_i:P(\bar{x})\To\w_1:P(\bar{x})$}
		\AxiomC{($\Gamma_{i+1}\To\Delta_{i+1}$)}
		\noLine
		\UnaryInfC{$\Gamma'_i,\ \s_i:P(\bar{x}),\ \w_1:P(\bar{x}),\ \dots,\ \w_m:P(\bar{x})\To\Delta_{i}$}
		\RightLabel{(\textsf{cut})}
		\BinaryInfC{$\Gamma'_i,\ \s_i:P(\bar{x}),\ \Gamma'_i,\ \s_i:P(\bar{x}),\ \w_2:P(\bar{x}),\ \dots,\ \w_m:P(\bar{x})\To\Delta_i$}
		\RightLabel{(\textsf{c}$\To$)}
		\UnaryInfC{$\Gamma'_i,\ \s_i:P(\bar{x}),\ \w_2:P(\bar{x}),\ \dots,\ \w_m:P(\bar{x})\To\Delta_i$}
		\RightLabel{(\textsf{cut})}
		\BinaryInfC{\phantom{i}$\dots$}
\end{prooftree}

\noindent
Now, we consider the case of a labelled formula $s_i:\phi_i$ with $\phi_i=\Box\psi$. By induction hypothesis, $\not\vdash\Gamma_i\To\Delta_i$, so $s_i:\phi_i$ cannot be in both $\Gamma_i$ and $\Delta_i$. We thus have three cases to consider.
\begin{itemize}
	\item Case 1: $s_i:\phi_i$ is neither in $\Gamma_i$ nor in $\Delta_i$. We simply set $\Gamma_{i+1}=\Gamma_i$, $\Delta_{i+1}=\Delta_i$.
	\item Case 2: $\s_i:\phi_i\in\Gamma_i$. Consider the set $\{\t\mid\text{for some }\w\in\s_i, \w\R\t\subseteq\Gamma_i\}$. This set is finite, since $\Gamma_i$ is finite. Let $\t_1,\dots,\t_k$ be all the elements. Now we define a sequence $(\Gamma_i^j)_{j\le k}$ by letting $\Gamma_i^0=\Gamma_i$ and $\Gamma_i^{j+1}=\Gamma_i^j\cup\{\t_{j+1}:\psi\}$. 
	
	Inductively, each sequent $\Gamma_i^j\To\Delta_i$ is not derivable. To see this, suppose $\Gamma_i^j\To\Delta_i$ is not derivable but $\Gamma_i^{j+1}\To\Delta_i$ is. The latter means that we can derive $\Gamma_i^j\cup\{\t_{j+1}:\psi\}\To\Delta_i$. However, since $\Gamma_i^j$ contains $\s_i:\Box\psi$ and also $\w\R\t_{j+1}$ for some $\w\in\s_i$, by the rule $(\Box{\To})$ we have that $\Gamma_i^j\To\Delta_i$ is derivable, contrary to assumption.
	
	Finally, we set $\Gamma_{i+1}=\Gamma_i^{k}=\Gamma_i\cup\{\t_j:\psi\mid 1\le j\le k\}$ and $\Delta_{i+1}=\Delta_i$.
	
	\item {Case 3: $\s_i:\phi_i\in\Delta_i$}. Note that since $\Gamma_i$ and $\Delta_i$ are finite, the corresponding universe $\W_{\Gamma_i,\Delta_i}$ is finite. Therefore, there are infinitely many natural numbers that are not in this universe. Let $\vv_1,\dots,\vv_{n_\psi}$ be the first $n_\psi$ such labels and let $\t=\{\vv_1,\dots,\vv_{n_\psi}\}$. Now consider for each $\w\in \s_i$ the corresponding sequent:
	$$\Gamma_i,\w\R\t\To\Delta_i,\t:\psi$$
	If all these sequents were derivable, then by the rule $({\To}\Box)$, $\Gamma_i\To\Delta_i,\s_i:\Box\psi$ would be derivable; and since $\s_i:\Box\psi\in\Delta_i$, by contraction, 
	$\Gamma_i\To\Delta_i$ would be derivable, contrary to assumption.  
	So, for at least one $\w\in\s_i$, the corresponding sequent is not derivable. Let $\w^*$ be the least such number and define $\Gamma_{i+1}=\Gamma_i\cup\w^*\R\t$ and $\Delta_{i+1}=\Delta_i\cup\{\t:\psi\}$.
	
\end{itemize}
Finally we set $\Gamma^+=\bigcup_i \Gamma_i$ and $\Delta^+=\bigcup_i\Delta_i$. Clearly, $\Gamma^+\To\Delta^+$ is not derivable: if that was the case, there would be some finite subsets $\Gamma'\subseteq\Gamma^+$ and $\Delta'\subseteq\Delta^+$ such that $\vdash \Gamma'\To\Delta'$. Due to their finiteness, there must be some $i\in\N$ such that $\Gamma'\subseteq\Gamma_i$ and $\Delta'\subseteq\Delta_i$. By weakening, this would imply $\vdash\Gamma_i\To\Delta_i$, which we know not to be the case. 
The fact that $\Gamma^+\To\Delta^+$ is saturated is straightforwardly ensured by the inductive construction. As an illustration, we show that the $(\Box L)$ condition is satisfied. 

Suppose that $\s:\Box\psi\in\Gamma$ and suppose that for some $\w\in\s$ we have $\w\R\t\subseteq \Gamma$ for some finite set $\t$. We need to show $\t:\psi\in\Gamma$.	Let $i$ be a number such that $\Gamma_i$ contains $\s:\Box\psi$ as well as all the relational atoms in the set $\w\R\t$ (such a number exists since $\t$ is finite and so $\w\R\t$ is a finite set of atoms). Let $j\ge i$ be a number that enumerates the labelled formula $\s:\Box\psi$ (which exists since we assumed that all labelled formulas are enumerated infinitely many times). Then since $\Gamma_j$ contains $\s:\Box\psi$ as well as $\w\R\t$, our procedure makes sure that $\Gamma_{j+1}$ (and so also $\Gamma^+$) contains $\t:\psi$, as required.
\end{proof}

\noindent
We have thus shown that any non-derivable sequent can be extended to a saturated $g$-sequent. From any such $g$-sequent, we can then construct a canonical model that refutes it. The saturation conditions will guarantee that this model acts as a countermodel for the original sequent.
\paragraph{Canonical model construction} Let $\Gamma\To\Delta$ be a saturated g-sequent. We define a canonical model $M_{\Gamma,\Delta}^c=\langle\W_{\Gamma,\Delta},D^c,R^c,I^c\rangle$ as follows:

\begin{itemize}
\item $\W_{\Gamma,\Delta}$ is the universe of the g-sequent $\Gamma\To\Delta$, i.e. the set of all indices that occur in it;
\item $D^c$ is the set of variables occurring in $\Gamma\To\Delta$;
\item $\w R^c \vv\iff \w\R\vv\in\Gamma$;
\item $\langle x_1,\dots,x_n\rangle\in I^c_\w(P)\iff \{\w\}:P(x_1,\dots,x_n)\in\Gamma$.
\end{itemize}

\begin{lemma}[Support lemma]\label{Support-Lemma}
For all labelled formulas $\s:\phi$,
\begin{itemize}
	\item if $\s:\phi\in\Gamma$ then $M_{\Gamma,\Delta}^c,\s\models_{\text{id}}\phi$
	\item if $\s:\phi\in\Delta$ then $M_{\Gamma,\Delta}^c,\s\not\models_{\text{id}}\phi$
\end{itemize}
\end{lemma}
\noindent\emph{Proof.}
By induction on $\phi$. We consider the cases of $\phi=P(\bar{x})$ and $\phi=\Box\psi$ (we omit the subscript for the assignment, for readability).
\begin{itemize}
	\item Suppose $\s:P(\bar{x})\in\Gamma$. By saturation (\textsf{at}$L$), for all $\w\in\s$ we have that $\w:P(\bar{x})\in\Gamma$. By definition of $M^c_{\Gamma,\Delta}$, for all $\w\in\s$: $\bar{x}\in I^c_w(P)$. Therefore, $M^c,\s\models P(\bar{x})$.
	\item Suppose $\s:P(\bar{x})\in\Delta$. By saturation (\textsf{at}$R$), for some $\w\in\s$, we have that $\w:P(\bar{x})\in\Delta$. Clearly, then, $\w:P(\bar{x})\notin\Gamma$, or we would have that $\Gamma\To\Delta$ is an instance of (\textsf{id}), contradicting (\textsf{unprov}). This implies that $\bar{x}\notin I^c_\w(P)$, meaning that $M^c,\s\nmodels P(\bar{x})$.
\item Suppose $\s:\Box\psi\in\Gamma$. We need to show that for any $\w\in\s$, $M^c_{\Gamma,\Delta},R^c[\w]\models\psi$. 
Since $\psi$ is finitely coherent, it suffices to show that for any $\w\in\s$ and any finite $\t\subseteq R^c[\w]$, $M^c_{\Gamma,\Delta},\t\models\psi$.

So, take a $\w\in\s$ and a finite $\t\subseteq R^c[\w]$. By definition of $R^c$, the latter means that for each $\vv\in\t$, $\w\R\vv\in\Gamma$. So, we have $\w\R\t\subseteq\Gamma$. By saturation, $\s:\Box\psi\in\Gamma$ and $\w\R\t\subseteq\Gamma$ together imply $\t:\psi\in\Gamma$. By induction hypothesis, $M^c_{\Gamma,\Delta},\t\models\psi$, as required.

\item Suppose $\s:\Box\psi\in\Delta$. By saturation, for some $\w\in\s$ and some finite set $\t$ we have $\w\R\t\subseteq\Gamma$ and $\t:\psi\in\Delta$. By induction hypothesis, this gives $M^c_{\Gamma,\Delta},\t\not\models\psi$. Since $\w\R\t\subseteq\Gamma$ we have $\t\subseteq R^c[\w]$. By persistency, $M^c_{\Gamma,\Delta},R^c[\w]\not\models\psi$. Since $\w\in\s$, we conclude $M^c_{\Gamma,\Delta},\s\not\models\Box\psi$.\hfill$\Box$
\end{itemize}
\noindent
We can now prove that \IWMC{} is complete in the sense that it derives any valid (finite) sequent.

\begin{proposition}[Completeness for sequents]\label{Weak-completeness}
For any sequent $\Gamma\To\Delta$ we have: $$\models\Gamma\To\Delta\quad\text{ implies }\quad
\vdash\Gamma\To\Delta$$
\end{proposition}
\begin{proof}
By contraposition. Assume that $\nvdash\Gamma\To\Delta$. Then, by Lemma \ref{Saturation-Lemma}, $\Gamma\To\Delta$ has a saturated extension $\Gamma^+\To\Delta^+$ such that $\nvdash\Gamma^+\To\Delta^+$. By Lemma \ref{Support-Lemma}, there is $M$, $f$, and $g$ such that $M,f\nVdash_g \Gamma^+\To\Delta^+$ which, in particular, implies $M,f\nVdash_g\Gamma\To\Delta$. 
\end{proof}

%

\noindent
We have thus established that \IWMC{} is complete with respect to sequents: it derives all and only the valid ones. Via the connection provided by Proposition \ref{prop:connection}, it follows that \IWMC{} is also complete (and indeed, by compactness, strongly complete) with respect to entailment in the logic \inqwqml. 




\begin{theorem}[Strong completeness for \inqwqml]
For any set of \inqwqml{} formulas $\Phi\cup\{\psi\}$, letting $\s_\psi=\{1,\dots,n_\psi\}$ where $n_\psi$ is the coherence estimate for $\psi$ as given by Prop.\ \ref{Finite-coherence}, we have:
$$\text{$\Phi\models\psi\quad$ iff $\quad\vdash\{\s_\psi:\phi\mid\phi\in\Phi\}\To \s_\psi:\psi$}$$
\end{theorem}
\begin{proof} 
	$(\implies):$ If $\Phi\models\psi$, by compactness (Proposition \ref{compactness}), there is some finite $\Phi_0\subseteq\Phi$ such that $\Phi_0\models\psi$. By the definition of $\vdash$, it suffices to prove that $\vdash\{\s_\psi:\phi\mid\phi\in\Phi_0\}\To\s_\psi:\psi$. Assume towards a contradiction that $\nvdash\{\s_\psi:\phi\mid\phi\in\Phi_0\}\To\s_\psi:\psi$. By the completeness result for sequents, this implies that $\nmodels\{\s_\psi:\phi\mid\phi\in\Phi_0\}\To\s_\psi:\psi$. Then, since $\Phi_0$ is finite and $\# \s_\psi=n_\psi$, by Proposition \ref{prop:connection} we have that $\Phi_0\nmodels\psi$, which contradicts the initial assumptions.
	
	$(\impliedby):$ By contraposition, assume that $\Phi\nmodels\psi$. Then, for all finite $\Phi_0\subseteq\Phi$, $\Phi_0\nmodels\psi$. By Proposition \ref{prop:connection}, since $\#\s_\psi=n_\psi$, this implies that for any such $\Phi_0$, $\nmodels \{\s_\psi:\phi\mid\phi\in\Phi_0\}\To\s_\psi:\psi$. 
	By the soundness result for sequents (Proposition \ref{soundness}), we have that for all finite $\Phi_0\subseteq\Phi$, $\nvdash \{\s_\psi:\phi\mid\phi\in\Phi_0\}\To\s_\psi:\psi$ which, by the definition of $\vdash$ combined with the admissibility of (${\To}\textsf{w}$), gives us $\nvdash \{\s_\psi:\phi\mid\phi\in\Phi\}\To\s_\psi:\psi$.

\end{proof}
\subsection{Adding identity.}\label{Section:identity}
In this section, we show how to generalize our approach to the case of signatures including the identity predicate. We start with a brief introduction to the implementation of identity in \inqwqml. For a more complete presentation of identity in first-order inquisitive logic, we refer to \cite[Section 5.4]{Ciardelli:23book}.
\paragraph{Identity in \inqwqml}
In first-order inquisitive logics, identity is seen as a binary predicate whose interpretation can change across possible worlds. Syntactically, we extend the definition of the language of \inqwqml{} to include, amongst the predicate atoms, identity atoms of the form $x=y$, where $x,y\in Var$.

Semantically, the identity predicate $=$ is interpreted as a standard binary predicate. Therefore, in a model $M=\langle W,D,R,I\rangle$, for any world $w\in W$, $I_w(=)\subseteq D^2$. However, the interpretation of the identity predicate must satisfy two additional constraints. For a predicate $P$ and world $w$, let $P_w$ and $=_w$ denote, respectively, $I_w(P)$ and $I_w(=)$. Then, for any $w\in W$: \begin{itemize}
	\item Congruence: for any $d_1,d'_1,\dots,d_n,d'_n\in D$ and $n$-ary predicate $P$, if $d_1=_w d'_1,\dots, d_n=_w d'_n$, then: $$\text{$\langle d_1,\dots,d_n\rangle\in P_w\iff \langle d'_1,\dots,d'_n\rangle\in P_w$}$$
	\item Equivalence: $=_w$ is an equivalence relation on $D$
\end{itemize}
\noindent
The recursive definition of the support relation remains the same as in the case without identity, and the clause for identity atoms can be obtained from the one for predicates:\begin{itemize}
	\item $M,s\models_g x=y\iff \text{for all $w\in s$, }\langle g(x),g(y)\rangle\in I_w(=)$
\end{itemize}
\noindent
\paragraph{Extending the proof system.}
Since we see identity as a binary predicate, we extend the applicability of $({\To}\textsf{at})$ and $(\textsf{id})$ to labelled identity atoms $\s:x=y$. However, we also need to extend \IWMC{} with rules encoding the specific properties of identity. We add two rules, inspired by \cite[Section 6.5]{Negri:01}, and call the extended calculus $\IWMC^=$. In the rules, $x,y\in Var$, $\bar z$ is a sequence of variables, and $\phi[[y/x]]$ denotes the substitution of an arbitrary number of instances of $x$ in $\phi$ with instances of $y$:\\[.3cm]
\begin{minipage}{0.49\textwidth}
\centering
$\infer[(\textsf{Ref})]{\Gamma\To\Delta}{\s:x=x,\ \Gamma\To\Delta}$
\end{minipage}
\begin{minipage}{.5\textwidth}
\centering
$\infer[(\textsf{Repl})]{\s:x= y,\ \s:P(\bar z),\ \Gamma\To\Delta}{\s:x= y,\ \s:P(\bar z),\ \s:P(\bar z)[[y/x]],\ \Gamma\To\Delta}$
\end{minipage}\\

\noindent
The extended system $\IWMC^=$ satisfies analyticity and the structural properties that we proved for \IWMC{} in Section \ref{subsection-structural}. Additionally, the following rules encoding the properties of $=$ as an equivalence relation are admissible:\\[.3cm]
{\begin{minipage}{0.4999\textwidth}
		\centering
		$\infer[(\textsf{Sym})]{\s:x= y,\ \Gamma\To\Delta}{\s:x= y,\ \s:y=x,\ \Gamma\To\Delta}$
\end{minipage}}
{\begin{minipage}{0.49\textwidth}
		\centering
		$\infer[(\textsf{Trans})]{\s:x= y,\ \s:y=z,\ \Gamma\To\Delta}{\s:x= y,\ \s:y=z,\ \s:x=z,\ \Gamma\To\Delta}$
\end{minipage}}\\

\noindent
 For any formula $\phi$, the replacement axiom $\s:x= y,\ \s:\phi\To\s:\phi[[y/x]]$ can be shown to be derivable by adapting the argument provided in \cite[Lemma 6.5.2]{Negri:01}. Our proof uses the admissibility of the following sub-label replacement rule:\\[.3cm]
	\begin{minipage}{.999\textwidth}
	\centering
	$\infer[(\textsf{Sub-Repl})\;\text{\small where }\s'\subseteq \s]{\s:x= y,\ \s':P(\bar z),\ \Gamma\To\Delta}{\s:x= y,\ \s':\phi,\ \s':P(\bar z)[[y/x]],\ \Gamma\To\Delta}$
\end{minipage}\\

\noindent
We give an informal sketch of the proof, which is mostly routine. We prove, by induction on the length of $\phi$, the stronger claim that $\s:x= y,\ \s':\phi\To\s':\phi[[y/x]]$ is derivable for any $s'\subseteq s$. The base case for $\phi$ atomic is immediate by the admissibility of (\textsf{Sub-Repl}). The stronger claim is needed to prove the inductive case of implication, where subsets of $s$ are generated by the rules for $\to$. Here, the structure of the derivation is the same as in the original proof, but the inductive hypotheses involve subsets of the initial label. The remaining cases, including the inductive case of $\Box$, follow from standard arguments.

Combining the structural properties of $\IWMC^=$ with the derivability of the replacement axiom and the admissibility of the above rules makes it straightforward to prove the admissibility of ($\textsf{G-Repl}$), which generalizes (\textsf{Repl}) to arbitrary formulas, and of (${\To}$\textsf{G-Repl}), a generalized right-replacement rule:\\[.3cm]
{\begin{minipage}{0.4999\textwidth}
		\centering
		$\infer[(\textsf{G-Repl})]{\s:x= y,\ \s:\phi,\ \Gamma\To\Delta}{\s:x= y,\ \s:\phi,\ \s:\phi[[y/x]],\ \Gamma\To\Delta}$
\end{minipage}}
{\begin{minipage}{0.49\textwidth}
	\centering
	$\infer[({\To}\textsf{G-Repl})]{\s:x= y,\ \Gamma\To\Delta,\ \s:\phi}{\s:x= y,\ \Gamma\To\Delta,\ \s:\phi,\ \s:\phi[[y/x]]}$
	\end{minipage}}\\[.3cm]

\noindent
Using analogous arguments, one can prove the admissibility of their sub-label versions (\textsf{Sub-G-Repl}) and (${\To}$\textsf{Sub-G-Repl}), where $\s:\phi$ and $\s:\phi[[y/x]]$ are replaced, respectively, by $\s':\phi$ and $\s':\phi[[y/x]]$, with the requirement that $\s'\subseteq\s$ (notice that the stronger claim used in the derivability proof for the replacement axiom above is exactly what is needed in this more general case).
\paragraph{Completeness proof.}
To adapt the definition of saturated sequent, we extend (\textsf{at}$L$) and (\textsf{at}$R$) to labelled identity atoms of the form $\s:x=y$ and we add the following clauses: 
\begin{itemize}
		\item[]$(= L)$ If $\s:x= y\in\Gamma$, then: (i) if $\s:\phi\in\Gamma$, then for any possible substitution $\phi[[y/x]]$, $\s:\phi[[y/x]]\in\Gamma$ and (ii) if $\s:\phi\in\Delta$, then for any possible substitution $\phi[[y/x]]$, $\s:\phi[[y/x]]\in\Delta$
		\item[]$(= Ref)$ For all $x\in Var$ that occur in $\Gamma\cup\Delta$, for all singleton labels $\w\in W_{\Gamma,\Delta}$, $\w:x= x\in\Gamma$
\end{itemize}\hspace*{0cm}
	
	\noindent
	The proof of the saturation lemma requires two changes. We fix an enumeration $(z_i)_{i\in\mathbb{N}}$ of all variables in the extended language. First, for each $i\in\N$, we add to $\Gamma_{i+1}$ the finite set of labelled formulas $\textsf{Id}(i)\coloneqq\{k:z_j= z_j\mid 0\leq k,j\leq i,\ k\in W_{\Gamma_i,\Delta_i} \text{, and }z_j\text{ occurs in }\Gamma_i\cup\Delta_i\}$. Underivability holds by (\textsf{Ref}). 
	We also extend the case of an atom $\s_i:P(\bar{x})$ being in $\Delta_i$ to include identity atoms of the form $\s_i:x=y$.
	
	Then, we include one additional inductive case for identity atoms. If $\phi_i$ is $x= y$ and $\s_i:x= y\in\Gamma_i$, we let $\Gamma_i^{\textsf{at}}=\Gamma_i\cup\{\w:x=y\mid\w\in\s_i\}$ and we let $Sub_i(\phi)$ be the set of all possible substitutions $\s':\phi[[y/x]]$ for all $\s'\subseteq\s_i$. We define $\Gamma_{i+1}=\Gamma_i^{\textsf{at}}\cup(\bigcup_{\s_i:\phi\in\Gamma_i^{\textsf{at}}}Sub_i(\phi))\ \cup\ \textsf{Id}(i)$. 
	and  $\Delta_{i+1}=\Delta_i\cup(\bigcup_{\s_i:\phi\in\Delta_i}Sub_i(\phi))$.  
	To prove that $\Gamma_i^{\textsf{at}}\To\Delta_i$ is underivable, we proceed as in the proof of Lemma \ref{Saturation-Lemma} for the atomic case. The underivability of $\Gamma_{i+1}\To\Delta_{i+1}$ is guaranteed by the rule $(\textsf{Ref})$ in the case of $\textsf{Id}(i)$ and by the rules $(\textsf{Sub-G-Repl})$ and $({\To}\textsf{Sub-G-Repl})$ for the remaining additions to $\Gamma_{i+1}$ and $\Delta_{i+1}$.  
	With this construction, verifying the saturation conditions is straightforward.\\[.2cm]	
	The definition of the canonical model remains unchanged from the case without identity, since we view	identity as a binary predicate. In particular, the derived definition for the extension of identity is $\langle x,y\rangle\in I^c_\w(=)\iff \w:x=y\in\Gamma$. The verification of the congruence property of $=$ is immediate thanks to the saturation conditions. The fact that $=$ is an equivalence relation also follows easily: for symmetry, if $\w:x=y\in\Gamma$, then $\w:x=x\in\Gamma$ by ($=Ref$) and, therefore, $\w:y=x\in\Gamma$ by ($=L$); for transitivity, if $\w:x=y,\ \w:y=z\in\Gamma$, then we have that $\w:y=x\in\Gamma$, which can be used in combination with $\w:y=z\in\Gamma$ to get that $\w:x=z\in\Gamma$ by ($=L$). 
		The remaining steps of the completeness proof are analogous to the identity-free case.

\section{Future work}\label{Section:conclusions}
	We conclude by outlining some directions for potential future extensions of this work. 
	
	A natural question is whether the approach can be extended to arbitrary signatures containing function symbols as well as the identity predicate. In inquisitive first-order logics, the interpretation of function symbols can vary between possible worlds. For this reason, there is a distinction between \textit{rigid} function symbols and terms, whose interpretation is constant, and \textit{non-rigid} ones. We denote the first in bold (e.g., \textbf{f, t}). These two classes of syntactic objects are known to satisfy different properties semantically (see \cite[Chapter 5]{Ciardelli:23book} for a detailed discussion). For signatures containing function symbols, it seems natural to assume that $\IWMC^=$ could be adapted by using arbitrary terms instead of variables in the rules $(\textsf{id})$ and $({\To}\textsf{at})$, and by replacing $(\forall{\To})$ with the following two rules, where $\alpha$ stands for a classical formula, $t$ for an arbitrary fresh term and $\textbf{t}$ for a fresh rigid term:\\[.3cm]
	\begin{minipage}{.49\textwidth}
\centering
$\infer[\rulename{\forall^\textsf{gen}{\To}}]
{\s:\forall x.\phi,\,\Gamma \To \Delta}
{\s:\phi[\mathbf{t}/x],\,\s:\forall x.\phi,\,\Gamma \To \Delta}$
\end{minipage}			
\begin{minipage}{.50\textwidth}
\centering
$\infer[\rulename{\forall^\textsf{cl}{\To}}]
{\s:\forall x.\alpha,\,\Gamma \To \Delta}
{\s:\alpha[t/x],\,\s:\forall x.\alpha,\,\Gamma \To \Delta}$
\end{minipage}\\

In Section \ref{subsection-structural}, we showed that \IWMC{} enjoys key structural properties. It remains an open question whether the properties of our calculus can be used to obtain metatheoretical results about the logic \inqwqml, such as interpolation. The structural properties we established, and in particular invertibility, may also allow us to define a systematic proof search procedure for our calculus; completeness could then be established by showing that a failed proof search always produces a countermodel, following a well-established strategy for labelled sequent calculi for modal logics (see, e.g., \cite{GargGenoveseNegri:12,Negri:14}).

Other inquisitive modalities are natural candidates for extending the present approach. A particularly interesting case is that of neighborhood-based inquisitive modalities (see, for instance, the ones proposed in \cite{CiardelliInqNL:25} and \cite{CiardelliRoelofsen:15idel}). In neighborhood frames, the accessibility relation associates to each possible world a set of sets of possible worlds. It would be interesting to explore whether our labelled sequent calculus can be adapted to this semantic setting with similarly positive results.
\bibliographystyle{eptcs}
\bibliography{aiml26-3}

\end{document}